\documentclass[onecolumn]{svjour3}             % onecolumn (second format)

\smartqed  % flush right qed marks, e.g. at end of proof
\usepackage{epsfig}
\usepackage{epstopdf}
\usepackage{graphics}
\usepackage{multirow}
\usepackage[usenames]{color}
\usepackage{lscape}
\usepackage{color}
\usepackage[ruled]{algorithm2e}

\begin{document}

\title{Depth-Optimized Reversible Circuit Synthesis}
\author{Mona~Arabzadeh, Morteza~Saheb~Zamani, Mehdi~Sedighi, Mehdi~Saeedi}
\authorrunning{M. Arabzadeh, M. Saheb Zamani, M. Sedighi, M. Saeedi} % if too long for running head

\institute{
A preliminary and partial version of this paper was presented at the 2011 International Workshop on Logic and Synthesis, San Diego, USA.\\\\
M. Arabzadeh, M. Saheb Zamani, M. Sedighi, M. Saeedi \at
Computer Engineering Department, Amirkabir University of Technology, Tehran, Iran.\\
\email{\{m.arabzadeh, szamani, msedighi, msaeedi\}@aut.ac.ir}             \\
M. Saeedi is currently with the Department of Electrical Engineering, University of Southern California, Los Angeles, CA, USA 90089-2562. \email{msaeedi@usc.edu}
}


\maketitle

\begin{abstract}%
In this paper, simultaneous reduction of circuit depth and synthesis cost of reversible circuits in quantum technologies with limited interaction is addressed.
We developed a cycle-based synthesis algorithm which uses negative controls and limited distance between gate lines. To improve circuit depth, a new parallel structure is introduced in which before synthesis a set of disjoint cycles are extracted from the input specification and distributed into some subsets. The cycles of each subset are synthesized independently on different sets of ancillae. Accordingly, each disjoint set can be synthesized by different synthesis methods.
Our analysis shows that the best worst-case synthesis cost of reversible circuits in the linear nearest neighbor architecture is improved by the proposed approach. Our experimental results reveal the effectiveness of the proposed approach to reduce cost and circuit depth for several benchmarks.


\end{abstract}
\section{Introduction} \label{sec:intro}
Boolean reversible circuits have attracted attention as components in several quantum algorithms including Shor's quantum factoring \cite{markov2012} and stabilizer circuits \cite{Aaronson&Gottesman04}. In the recent years, considerable efforts have been made to synthesize a Boolean reversible function by a set of quantum gates \cite{saeedi2012}.

The proposed technologies for quantum computing
suffer from practical limitations for implementation.
For example, popular quantum technologies allow computation on a few qubits in a
linear nearest neighbor (LNN) architecture where only adjacent qubits can interact \cite{cheung07}.
Additionally, physical qubits are fragile and can hold their states only for a limited time, called \emph{coherence time}, \cite{meter06}.
To reflect technological constraints in
the synthesis stage, different technology-specific cost metrics have
been introduced.

\begin{itemize}
  \item \emph{Two-qubit cost} is the number of two-qubit gates of any type and the number of one-qubit gates (reported separately) in a given circuit. The number of two-qubit gates for an $n$-qubit Toffoli gate (for $n$ $\geq$ 3) is estimated as $10n-25$ \cite{maslov11}. \emph{Quantum cost} (QC) is the number of NOT, CNOT, controlled-V and controlled-V$^\dag$ gates required to implement a given reversible function.
  \item \emph{Interaction cost} is the distance between gate qubits for any two-qubit gate. Quantum circuit technologies with 1D, 2D and 3D interactions exist \cite{cheung07}. Interaction cost for a circuit is calculated by a summation over the interaction costs of its gates.
  \item \emph{Number of ancillae} and \emph{garbage qubits} reflect the limited number of qubits in the current quantum technologies.
  \item \emph{Depth} is the largest number of elementary gates on any path from inputs to outputs in a circuit. Reducing circuit depth can increase coherence time.

\end{itemize}

Synthesis of reversible Boolean circuits has an exponential search space. Consequently, many heuristic algorithms have been proposed to consider the effects of quantum cost and two-qubit cost in the synthesis stage [7-10].
%\cite{Gupta06,Maslov07,WilleDAC09,saeedi10}.
Additionally, several post-process optimization methods have been developed to improve quantum cost \cite{Maslov07,miller10,maslov11}, interaction cost \cite{saeedi11QIP,hirata2011}, and depth \cite{MaslovTCAD08}. However, the number of algorithms which consider different parameters simultaneously --- the focus of this work --- is very limited.

Besides technological limitations, studying theoretical aspects of circuits with either limited interactions among qubits of gates or limited depth attracts interest in complexity theory. %NC$^i$ is the class of decision problems solvable by a uniform family of Boolean circuits with polynomial size, depth of $O(\log^i n)$ and fan-in equal to two.
For example, NC$^i$ is the class of decision problems solvable by a uniform family of Boolean circuits with polynomial size, depth of $O(\log^i n)$ and fan-in=2. QNC is the class of constant-depth quantum circuits without fanout gates \cite{moore09}.

In this paper, a synthesis algorithm for Boolean reversible circuits is proposed which uses a cycle-based strategy to synthesize circuits for the LNN architecture. The proposed technique leads to improved synthesis costs as compared to the best prior methods for several benchmarks.
Moreover, a parallel structure for reversible Boolean circuits is presented which significantly reduces circuit depth with 2$n$ ancillae. Overall, our circuits can be considered as depth-optimized reversible circuits for the LNN architecture.

This paper is organized as follows.
Basic concepts are introduced in Section \ref{sec:pre}. Related synthesis and post-process optimization methods are reviewed in Section \ref{sec:relat}. The proposed cycle-based synthesis algorithm for the LNN architecture is described in Section \ref{sec:Mk-cycle}. Section \ref{subsec:parallel} presents a parallel structure to reduce circuit depth. Experimental results are reported in Section \ref{sec:exper}, and Section \ref{sec:conclu} concludes the paper.

%%%%%%%%%%%%%%%%%%%%%%%%%%%%%%%%%%%%%%%%%%%%%%%%%%%%%%%%%%%%%%%%%%%%%%%%%%%%%%%%%%%%%%%%%%%%%%%%%%%%%%%%%%%%%%%%%%%%%%%%%%%%%%%%%%%%%%%%%%%%%%%%%%%%%%%%%%%%%%%%%

\section{Basic Concepts} \label{sec:pre}
In this section, preliminary concepts are briefly introduced. Further background can be found in \cite{saeedi2012}.

\textbf{Permutation Function.} Let $B$ be any set and define $f$$:$ $B$ $\rightarrow$ $B$ as a one-to-one and onto transition function. The function $f$ is a \emph{permutation} function, as applying $f$ to $B$ leads to a set with the same elements of $B$ and probably in a different order. If $B =\{1, 2, 3, . . ., m\}$, there exist two elements $b_i$ and $b_j$ belonging to $B$ such that $f(b_i) = b_j$. A $k$-cycle with \emph{length} $k$ is denoted as $(b_1, b_2, ..., b_k)$ which means that $f(b_1) = b_2, f(b_2) = b_3,...,$ and $f(b_k) = b_1$. A given $k$-cycle $(b_1, b_2, . . . ,b_k)$ could be written in different ways, such as $(b_2, b_3, . . . b_k, b_1)$. Cycles $c_1$ and $c_2$ are called \emph{disjoint} if they have no common members. Any permutation can be written uniquely, except for the order, as a product of disjoint cycles. If two cycles $c_1$ and $c_2$ are disjoint, they can \emph{commute}, i.e., $c_1c_2$ = $c_2c_1$. A cycle with length two is called \emph{transposition}. A cycle or a permutation is called \emph{even} (\emph{odd}) if it can be written as an even (odd) number of transpositions. When $k$-cycle is even (odd) then $k$ is odd (even).

%%=============================================================================================================================================================

\textbf{Reversible Function.} An $n$-input, $n$-output, fully specified Boolean function $f$$:$ $B$ $\rightarrow$ $B$ over variables $X = \{ x_0,...,x_{n-1}\}$ is called \emph{reversible} if it maps each input pattern to a unique output pattern. Each reversible function can be considered as a permutation function.
The added lines to a circuit are called \emph{ancillae} and typically start out with a 0 or 1.

%%=============================================================================================================================================================

\textbf{Reversible Gate.}
An $n$-input, $n$-output gate is reversible if it realizes a reversible function.
A multiple-control Toffoli gate can be written as C$^m$NOT(C; t), where $C = \{i_1, \dots , i_m\}$ is the set of \emph{control} lines, $t =\{j\}$ with $C \cap t = \emptyset$ is the \emph{target} line and $0\leq i,j\leq {n-1}$.
A control line may be \emph{positive} (\emph{negative}) which means that if its value is one (zero), the value of the target is inverted. For $m$=0 and $m$=1, the gates are called NOT (N) and CNOT (C), respectively. For $m$=2, the gate is called C$^2$NOT or Toffoli (T).
The SWAP($a$,$b$) gate changes the value of two qubits $a$ and $b$, and can be constructed by three CNOT gates C($a$,$b$)C($b$,$a$)C($a$,$b$).
The controlled-V (controlled-V$^\dag$) gate changes the value of its target line using the transformation given by the matrix V (V$^\dag$) if the control line has the value 1.
\[
\small
V = \frac{{1 + i}}{2}\left[ {\begin{array}{*{20}c}
   1 & { - i}  \\
   { - i} & 1  \\
\end{array}} \right],V^ \dag   = \frac{{1 - i}}{2}\left[ {\begin{array}{*{20}c}
   1 & i  \\
   i & 1  \\
\end{array}} \right]
\]
\normalsize

%%=============================================================================================================================================================
%%=============================================================================================================================================================
%%=============================================================================================================================================================
%%=============================================================================================================================================================
\section{Related Work} \label{sec:relat}
%\subsection{Cycle-Based Synthesis Methods}
In this section, we review prior synthesis and optimization techniques that are used in this paper.

In \cite{Shende03}, an NCT-based synthesis method is proposed which decomposes a given
cycle into a set of transpositions. To implement an arbitrary transposition $(a,b) (c,d)$ for distinct $a$, $b$, $c$, $d$ $\neq$ 0, $2^i$, the authors introduced three subcircuits, namely $\pi$, $\kappa_0$ and $\pi^{-1}$ (the inverse of $\pi$), where the $\kappa_0$ circuit, C$^{n-2}$NOT($a_2,...,a_{n-1};a_0$), implements a fixed transposition ($2^n-4$, $2^n-3$) ($2^n-2$, $2^n-1$). Accordingly, a synthesis algorithm was proposed to transform $a$, $b$, $c$ and $d$ to $2^n-4$, $2^n-3$, $2^n-2$ and $2^n-1$, respectively. By cascading $\pi$, $\kappa_0$ and $\pi^{-1}$, an arbitrary transposition can be implemented with quantum cost $34n-64$.

The NCT-based synthesis method in \cite{Shende03} was extensively improved in \cite{saeedi10},
$k$-cycle method hereafter. In the $k$-cycle method, a given cycle of length $\geq 6$ is decomposed into a set of cycles of lengths $<6$, called \emph{elementary cycles}. Next, a set of synthesis algorithms was proposed to synthesize different elementary cycles, i.e., a pair of 2-cycles, a single 3-cycle, a pair of 3-cycles, a single 5-cycle, a pair of 5-cycles, a single 2-cycle (4-cycle) followed by a single 4-cycle (2-cycle) and a pair of 4-cycles. Similar to \cite{Shende03}, 0 and $2^i$ terms are fixed before synthesis
because their effect on their synthesis results is negligible \cite{saeedi10}. NCT gates with positive controls are used in both \cite{Shende03} and \cite{saeedi10}.
The effect of decomposition on the result of \cite{saeedi10} was considered in \cite{SaeediMEJ10} where a cycle-assignment technique based on graph matching was proposed. The worst-case quantum cost for synthesizing an arbitrary reversible function on $n$ lines is $8.5n2^n+o(2^n)$ in \cite{saeedi10}.

In \cite{MaslovTCAD08}, the authors introduced a post-process optimization algorithm to reduce the depth of a given quantum circuit. To achieve this, a set of circuit templates (circuit identities) was proposed to reduce quantum cost and circuit depth. The suggested templates are applied to change either gate locations or control/target positions in a subcircuit to parallelize more gates.
The introduced templates were used by a greedy algorithm which starts from gate $i$ and traverses the gates afterwards. At each step, the algorithm moves gates to left whenever possible and applies templates to check whether other gates can be moved to left or not. If no change is possible, it starts the same process from gate $i+1$.

%=============================================================================================================================================================
In \cite{saeedi11QIP}, a synthesis flow was proposed to improve the interaction cost of a given quantum circuit. The authors studied the exact synthesis of some small gates for the LNN architecture. The proposed optimal circuits are used to simplify larger circuits. Besides, some circuit templates are introduced to reduce the number of SWAP gates. Finally, local and global reordering of input qubits are considered to reorder gate qubits for improving the interaction cost. The proposed techniques were consolidated in a unified design flow to implement a given circuit with arbitrary interactions for architectures with limited interactions.

Fig. \ref{fig:examp1}-a shows a 3-input full adder with depth 4 \cite{MaslovTCAD08} and six elementary gates. Actually, depth 4 is optimal since four qubits are involved in the fourth qubit \cite{MaslovTCAD08}. Fig. \ref{fig:examp1}-b shows the same circuit after inserting SWAP gates to make the gate qubits adjacent with QC=24 and depth=23. Fig. \ref{fig:examp1}-c illustrates the same circuit after applying the method in \cite{saeedi11QIP} for reducing the number of SWAP gates where QC=18 and depth=17.

\begin{figure}
%\centering
\centerline{\includegraphics{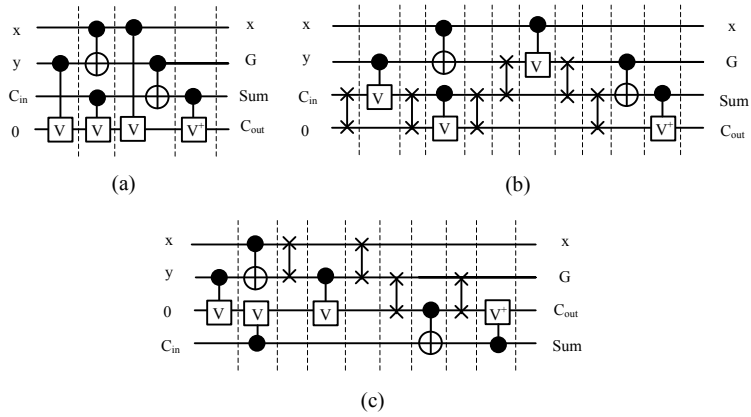}}
\caption{\small(a) 3-input reversible full adder with optimal depth 4 \cite{MaslovTCAD08}, (b) the circuit in (a) after inserting SWAP gates and (c) reducing the number of SWAP gates by \cite{saeedi11QIP}.}
\label{fig:examp1}
\end{figure}

%%%%%%%%%%%%%%%%%%%%%%%%%%%%%%%%%%%%%%%%%%%%%%%%%%%%%%%%%%%%%%%%%%%%%%%%%%%%%%%%%%%%%%%%%%%%%%%%%%%%%%%%%%%%%%%%%%%%%%%%%%%%%%%%%%%%%%%%%%%%%%%%%%%%%%%%%%%%%%%%%
\section{The Proposed Cycle-Based Synthesis Method for Interaction Cost} \label{sec:Mk-cycle}
The main contribution of \cite{saeedi10} is to propose a cycle-based synthesis approach with the primary focus
on quantum cost as the sole metric considered. However, another important implementational constraint, namely interaction cost, is considered besides the quantum cost in our proposed cycle-based method in this section.
To do that, we improve the $k$-cycle method by using negative controls and adapting the synthesis algorithms of elementary cycles to the LNN architecture.
Particularly, two new elementary odd cycles, a 2-cycle and a 4-cycle, are included to improve quantum cost. These odd cycles are synthesized as a pair of 2-cycles and a pair of 4-cycles in \cite{saeedi10} with one ancilla. Odd cycles need one ancilla in the NCT library for the implementation \cite{Shende03}. In our experiments, we used this ancilla for the decomposition of complex gates into elementary gates.
Additionally, 0 and $2^i$ terms are not fixed before synthesis to be used in the proposed parallel structure as discussed in Section \ref{subsec:parallel}.
%%=============================================================================================================================================================

\textbf{Negative controls} can reduce the number of elementary gates in the $\kappa_0$, $\pi$ and $\pi^{-1}$ circuits both with and without considering nearest neighbor restriction. Multiple-control Toffoli gates with at least one positive control can be simulated as efficiently as complex Toffoli gates with only positive controls \cite{MaslovTCAD08}.
By using CNOT and Toffoli gates with negative controls, one may not fix 0 and $2^i$ terms before synthesis as compared with the methods in \cite{Shende03,saeedi10}.

%%=============================================================================================================================================================

\textbf{Cycle Construction Length (CCL)} is defined as the number of lines required to implement a given cycle of length $L$. In theory, the minimum CCL is $\log_2L$. To implement the elementary cycles by NCT gates, at most two more lines are required in the proposed approach --- one to avoid Toffoli gates without any positive control in the $\kappa_0$ circuit, and one to improve circuit cost in the $\pi$, $\pi^{-1}$ circuits. Accordingly, we set CCL$_{(2)}$=2, CCL$_{(2,2)}$=4, CCL$_{(3)}$=3, CCL$_{(3,3)}$=5, CCL$_{(4)}$=4, CCL$_{(4,2)}$=5, CCL$_{(4,4)}$=5, CCL$_{(5)}$=5 and CCL$_{(5,5)}$=6.  For an $n$-line circuit, lines required to construct a given cycle, CCL in total, can be selected in $n$ $\times$ $(n-1)$ $\times$ ... $\times$$(n-\rm{CCL}-1)$ different ways. To improve interaction cost and depth
we place the selected lines close to each other in the middle of the $\kappa_0$ circuit at positions $k$, $k \pm 1$ and $ k \pm 2$ for $k=\lfloor n/2 \rfloor$. Details are discussed later.

\begin{figure*}
\centerline{\includegraphics[width=5in]{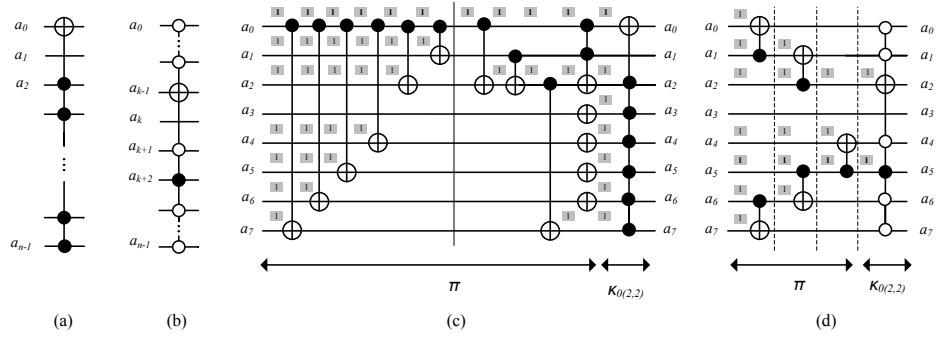}}
\caption{\small(a) The ${\kappa_{0(2,2)}}$ circuit in \cite{Shende03,saeedi10}. (b) The proposed ${\kappa_{0(2,2)}}$ circuit. Each control at position $i$, $0 \leq  i \leq n-1, i \neq k+2$ is negative. (c) An example of $\pi$ circuit in \cite{saeedi10}. $a_0$ is used to control CNOTs in the first part. The second subcircuit is the circuit in \cite[Theorem 3.1]{saeedi10}. (d) An example of $\pi$ circuit in the proposed method. Here, $k$=$3$. Refer to Table \ref{t:EC}.}
\label{fig:interval}
\end{figure*}

To synthesize a given elementary cycle, one needs to change input terms into the terms specified by the $\kappa_0$ circuit. This is done by converting the input terms into \emph{intermediate terms} specified by the $\pi$ circuit. Afterwards, the intermediate terms are transformed into $\kappa_0$ terms by a few specific gates, called \emph{static gates}. In the proposed method, the control and target lines in the $\pi$ circuit are selected such that interaction cost can be reduced. Since $\kappa_0$ cycles are constructed in the middle of the circuit and the intermediate terms are designed with at least one ``1'', as boldfaced in column \emph{Int. Terms} in Table \ref{t:EC}, it is possible to select control and target lines of each gate with length $ \leq \lceil(n-\rm{CCL})/2+\rm{CCL} \rceil$. Considering two SWAP gates with cost 6 leads to QC$_{LNN}$ $\leq 3(n+\rm{CCL})$ for each gate. To reduce circuit depth, the gates required to fix bit positions at the first half and the second half are applied in parallel. Algorithm \ref{alg:one} provides the details.

%%^^^^^^^^^^^^^^^^^^^^^^^^^^^^^^^^^^^^^^^^^^^^^^^^^^^^^^^^^^^^^^^^^^^^^^^^^^^^^^^^^^^^^^^^^^^^^^^^^^^^^^^^^^^^^^^^^^^^^^^^^^^^^^^^^^^^^^^^^^^^^^^^^^^^^^^^^^^^^
%\input{alg1.tex}
% Algorithm
\begin{algorithm}[t]
\SetAlgoNoLine
\KwIn{\\
\Indp $L$ $n$-bit input terms. Bit value at position $i$ of the $j$-th input term is $b_{(i,j)}$.\\
$L$ $n$-bit $\kappa_0$ terms. Bit value at position $i$ of the $j$-th $\kappa_0$ term is $b^{\kappa_0}_{(i,j)}$.\\
$Pivot$ is the boldfaced position in the intermediate terms in Table \ref{t:EC}.
}
%Cycle length ($L$), number of inputs/outputs ($n$), input terms ($b_{i,j}$ is the $j$-$th$ position of term $i$ in input terms), $\kappa_0$ terms ($b^{\kappa_0}_{i,j}$ is $j$-$th$ position of term $i$ in $\kappa_0$ terms), bold position in Table \ref{tab:EC} ($b$).}
\KwOut{The $\pi$ circuit.}

    \For{ $i$ in 0 to $L$}
    {
    	\If{$b_{(i,Pivot)}$$\neq$$1$}
    	{
            Set $b_{(i,Pivot)}$=1 by either a CNOT or a Toffoli gate;
        }

        \For{ $j$ in 0 to $Pivot$}
        {
            \If {$b_{(i,j)}$$\neq$ $b^{\kappa_0}_{(i,j)}$}
            {
                Find a position $p$: $b_{(i,p)}$=1, $b_{(k,p)}$$\neq$1 ($k<i$), $|p-j|$ is the minimum possible value, and $p\leq Pivot$;\\
              	Apply CNOT($p$;$j$);
            }

        }

        \For{ $j$ in $n$-$1$ to $Pivot$+$1$}
	 {
	     \If {$b_{(i,j)}$$\neq$ $b^{\kappa_0}_{(i,j)}$}
	     {
		Find a position $p$: $b_{(i,p)}$=1, $b_{(k,p)}$$\neq$1 ($k<i$), $|p-j|$ is the minimum possible value, and $p\geq Pivot$;\\
             	Apply CNOT($p$;$j$);
             }
         }
    }
\caption{Gate selection in the $\pi$ circuit}
\label{alg:one}
\end{algorithm}
%%^^^^^^^^^^^^^^^^^^^^^^^^^^^^^^^^^^^^^^^^^^^^^^^^^^^^^^^^^^^^^^^^^^^^^^^^^^^^^^^^^^^^^^^^^^^^^^^^^^^^^^^^^^^^^^^^^^^^^^^^^^^^^^^^^^^^^^^^^^^^^^^^^^^^^^^^^^^^^

The $\kappa_{0(2,2)}$ circuits in \cite{Shende03,saeedi10} and the proposed $\kappa_{0(2,2)}$ circuit are shown in Fig. \ref{fig:interval}-a and Fig. \ref{fig:interval}-b, respectively. Fig. \ref{fig:interval}-c illustrates one example of the $\pi$ circuit in \cite{saeedi10}. The input term is ``11110111'' which should be changed to the second term in the $\kappa_{0(2,2)}$ circuit in \cite{saeedi10}, i.e., ``11111101''. This is done by a circuit with QC=16 and depth=11. In contrast, ``11110111'' should be changed to ``00100100'' in the proposed method. Fig. \ref{fig:interval}-d shows the $\pi$ circuit with QC=5 and depth=3 based on Algorithm 1.

%%=============================================================================================================================================================

\subsection{Building Blocks}\label{subsubsec:BB}
In this section, direct synthesis of the suggested elementary cycles, i.e., (2), (2,2), (3), (3,3), (4,2), (4,4), (5), (5,5), is discussed.
Fig. \ref{fig:all-k} illustrates the $\kappa_0$ circuits of all elementary cycles. We give a full description of the synthesis method for a pair of 2-cycles first.

\textbf{(2,2)-synthesis:} To change $(a,b)(c,d)$ to $\kappa_{0(2,2)}$ terms:
\begin{itemize}
  \item At most $n$ NOT gates can be used to convert $a$ to ``0...1000...0''. Other terms $b$, $c$, and $d$ may be changed to new terms $b'$, $c'$ and $d'$, respectively.
  \item At most one CNOT gate conditioned on either the $i$-$th$ line $i$$\neq$$k+2$ (positive) or $i=k+2$ (negative) can be used to set the ($k-1$)-$th$ bit of $b'$. Next, at most $n-1$ CNOT gates conditioned on the ($k-1$)-$th$ bit can be applied to change the $j$-$th$ bit of $b'$ ($0 \leq j \leq n-1$, $j$$\neq$$k-1$) to ``0...1001...0''. $c'$, and $d'$ may be changed to new terms $c''$ and $d''$.
  \item At most one CNOT gate conditioned on either the $i$-$th$ line $i$$\neq$$k+2$ (positive) or $i=k+2$ (negative) can be used to set the $k$-$th$ bit of $c''$. Next, at most $n-1$ CNOT gates with positive control conditioned on the $k$-$th$ bit can be applied to change the $j$-$th$ bit of $c''$ ($0 \leq j \leq n-1$, $j\neq k$) to ``0...1010...0''. The last term $d''$ may be changed to a new term $d'''$.
  \item At most one CNOT gate conditioned on either the $i$-$th$ line $i$$\neq$$k+2$ (positive) or $i=k+2$ (negative) can be used to set the ($k+1$)-$th$ bit of $d'''$. Next, at most $n-1$ CNOT gates with positive control conditioned on the ($k+1$)-$th$ bit can be applied to change the $j$-$th$ bit of $d'''$ ($0 \leq j \leq n-1$, $j$$\neq$ $k+2$) to ``0...1111...0''.
  \item A Toffoli gate conditioned on the ($k-1$)-$th$ and the $k$-$th$ lines can be used to set the ($k+1$)-$th$ line. Therefore, it changes ``0...1111...0'' to ``0...1011...0''.
\end{itemize}
Note that converting each term does not corrupt the previously fixed terms. The same number of gates are needed for the $\pi^{-1}$ circuit. Accordingly, a total number of $8n+22$ elementary gates are required for the $\pi$ and $\pi^{-1}$ circuits.
The $\kappa_0$ circuit in Fig. \ref{fig:all-k}-b implements $(2^{k + 2} ,2^{k + 2}  + 2^{k - 1} )(2^{k + 2}  + 2^k ,2^{k + 2}  + 2^k  + 2^{k - 1} )$ with cost $24n-88$. Therefore, an arbitrary pair of 2-cycles $(a,b)(c,d)$ can be implemented by at most $32n-66$ elementary gates.

%\input{EC.tex}
%%^^^^^^^^^^^^^^^^^^^^^^^^^^^^^^^^^^^^^^^^^^^^^^^^^^^^^^^^^^^^^^^^^^^^^^^^^^^^^^^^^^^^^^^^^^^^^^^^^^^^^^^^^^^^^^^^^^^^^^^^^^^^^^^^^^^^^^^^^^^^^^^^^^^^^^^^^^
\begin{table*}
\caption{\small Direct synthesis of elementary cycles. Subscripts in the input cycles denote the orders in considering each term. The underlined bit in the $k$-$th$ position ($k$=$\lfloor \frac{n}{2}\rfloor$). The boldfaced ``1'' is Pivot in Algorithm 1. Numbers given for $\kappa_{0}$ terms are bit positions with ``1'' in the binary expansion.}
\centering{
\label{t:EC}
%\scriptsize
\begin{tabular}{|l|l|l|l|l|c|}
\hline
\multirow{2}{*}{Input Cycle(s)} &\multicolumn{3}{c|}{$\pi$ or $\pi^{-1}$ Circuit}&\multicolumn{2}{c|}{$\kappa_0$ Circuit}\\
 &Int. Terms&Max. Cost&Static Gates& Terms&Fig.\\

\hline
\hline
\multirow{2}{*}{$(a_1,b_2)$}		&(0...\textbf{1}\underline{0}...0)						&$n$ N				& -					&$(k + 1)$					 &\multirow{2}{*}{\ref{fig:all-k}-a}\\
					&(0...1\textbf{1}...0)				 				&$n$(1,$n$-1) C 		 & -					&$(k - 1)(k + 1)$				&				 \\
\hline
					&(0...\textbf{1}0\underline{0}0...0)						&$n$ N				 & -					&$(k + 2)$					 &\multirow{4}{*}{\ref{fig:all-k}-b}\\
	$(a_1,b_2)$			&(0...100\textbf{1}...0)							&$n$(1,$n$-1) C		  	 & -					&$(k + 2)(k - 1)$				&				 \\
	$(c_3,d_4)$			&(0...10\textbf{1}0...0)							&$n$(1,$n$-1) C			 & -					&$(k + 2)(k)$					&				 \\
					&(0...1\textbf{1}11...0)							&$n$(1,$n$-1) C			 &T$(k - 1 ,k ;k + 1 )$			&$(k + 2)(k)(k - 1)$ 				 &				\\
\hline
					&(0...0\underline{0}\textbf{1}...0)						&$n$ N				 & -					&$(k - 1)$					 &\multirow{3}{*}{\ref{fig:all-k}-c}\\
	$(a_1,b_2,c_3)$			&(0...\textbf{1}01...0)								&$n$(1,$n$-1) C 		 & -					&$(k+1)(k - 1)$					 &				 \\
					&(0...1\textbf{1}1...0)								&$n$(1,$n$-1) C		 	 & -					&$(k + 1)(k)(k - 1)$				&				 \\
\hline
					&(0...00\underline{0}0\textbf{1}...0)						&$n$ N				 & -					&$(k - 2)$ 					 &\multirow{6}{*}{\ref{fig:all-k}-d}\\
					&(0...000\textbf{1}1...0)							&$n$(1,$n$-1) C		 	 & -					&$(k - 1)(k - 2)$				&				 \\
	$(a_1,b_2,c_3)$			&(0...00\textbf{1}11...0)							&$n$(1,$n$-1) C			 & -					&$(k)(k - 1)(k - 2)$				 &				\\
	$(d_4,e_5,f_6)$			&(0...\textbf{1}0001...0)							&1 T, $n$-1 C			 & -					&$(k + 2)(k - 2)$				 &				 \\
					&(0...1\textbf{1}011...0)							&1 T, $n$-1 C			 &T$(k - 1 ,k + 2 ;k + 1 )$	 	&$(k + 2)(k - 1)(k - 2)$			 &				\\
					&(0...1\textbf{1}111...0)							&1 T, $n$-1 C			 &T$(k ,k + 2 ;k + 1 )$			&$(k + 2)(k)(k - 1)(k - 2)$ 			 &				\\

\hline
\multirow{4}{*}{$(a_1,b_2,c_3,d_4)$}	&(0...\textbf{1}\underline{0}00...0)						 &$n$ N				& -					&$(k+2)$					 &\multirow{4}{*}{\ref{fig:all-k}-e}\\
					&(0...100\textbf{1}...0)							&$n$(1,$n$-1) C 		 & -					&$(k - 1)(k + 2)$				&					 \\
					&(0...10\textbf{1}0...0)							&$n$(1,$n$-1) C			 & -					&$(k)(k + 2)$					&					 \\
					&(0...1\textbf{1}11...0)							&$n$(1,$n$-1) C			 & T$(k - 1 ,k ;k + 1 )$			&$(k - 1)(k)(k + 2)$				 &					\\
\hline
					&(0...\textbf{1}0\underline{0}00...0)						&$n$ N				 & -					&$(k + 2)$					 &\multirow{6}{*}{\ref{fig:all-k}-f}\\
					&(0...100\textbf{1}0...0)							&$n$(1,$n$-1) C 		 & -					&$(k + 2)(k - 1)$				&			 		 \\
	$(a_1,b_2,c_3,d_4)$		&(0...10\textbf{1}00...0)							&$n$(1,$n$-1) C			 & -					&$(k + 2)(k)$					 &					 \\
	$(e_5,f_6)$			&(0...1\textbf{1}110...0)							&$n$(1,$n$-1) C			 &T$(k - 1 ,k ;k + 1 )$ 	 		&$(k + 2)(k)(k - 1)$				 &					\\
					&(0...1011\textbf{1}...0)							&$n$(1,$n$-1) C			 & -					&$(k + 2) (k) (k - 1)(k - 2)$			 &					 \\
					&(0...1\textbf{1}011...0)							&1 T, $n$-1 C			 &T$(k - 2 ,k' ;k + 1 )$			&$(k + 2) (k - 1) (k - 2)$			 &					\\

\hline
					&(0...\textbf{1}0\underline{0}00...0)						&$n$ N				 & -					&$(k + 2)$					 &\multirow{8}{*}{\ref{fig:all-k}-g}\\
					&(0...1000\textbf{1}...0)							&$n$(1,$n$-1) C			 & -					&$(k + 2)(k - 2)$				&					 \\
					&(0...100\textbf{1}0...0)							&$n$(1,$n$-1) C			 & -					&$(k + 2)(k - 1)$				&					 \\
	$(a_1,b_2,c_3,d_4)$		&(0...10\textbf{1}11...0)							&$n$(1,$n$-1) C			 & - 					&$(k + 2)(k - 1)(k - 2)$			 &					\\
	$(e_5,f_6,g_7,h_8)$		&(0...10\textbf{1}00...0)							&$n$(1,$n$-1) C			 & T$(k - 2 ,k - 1 ;k )$	 		&$(k + 2) (k)$					 &					\\
					&(0...1\textbf{1}101...0)							&1 T, $n$-1 C			 & T$(k - 2 ,k ;k + 1 )$			&$(k + 2)(k) (k - 2)$				 &					\\
					&(0...1\textbf{1}110...0)							&1 T, $n$-1 C			 & T$(k - 1 ,k ;k + 1 )$			&$(k + 2) (k) (k - 1)$				 &					\\
					&(0...1\textbf{1}111...0)							&$n$(1,$n$-1) C			 & T$(k - 2 ,k - 1 ,k ;k + 1 )$		&$(k + 2) (k)  (k - 1) (k - 2)$			&					 \\

\hline
					&(0...\textbf{1}0\underline{0}00...0)						&$n$ N				 & -					&$(k + 2)$					 &\multirow{5}{*}{\ref{fig:all-k}-h}\\
					&(0...1000\textbf{1}...0)							&$n$(1,$n$-1) C			 & -					&$(k + 2)(k - 2)$				&					 \\
	$(a_1,b_4,c_2,d_3,e_5)$		&(0...100\textbf{1}0...0)							 &$n$(1,$n$-1) C			& -					&$(k + 2)(k - 1)$				 &					 \\
					&(0...1\textbf{1}011...0)							&$n$(1,$n$-1) C			 &T$(k - 2 ,k - 1 ;k + 1 )$ 		&$(k + 2) (k - 1)(k - 2)$			 &					\\
					&(0...10\textbf{1}11...0)							&$n$(1,$n$-1) C			 & -					&$(k + 2)(k)(k - 1)(k - 2)$			&					 \\
\hline
					&(0...\textbf{1}0\underline{0}000...0)						&$n$ N				 & -					&$(k + 2)$					 &\multirow{10}{*}{\ref{fig:all-k}-i}\\
					&(0...10000\textbf{1}...0)							&$n$(1,$n$-1) C			 & -					&$(k + 2)(k - 3)$				&					 \\
					&(0...1000\textbf{1}0...0)							&$n$(1,$n$-1) C			 & 					&$(k + 2)(k - 2)$				&					 \\
					&(0...100\textbf{1}11...0)							&$n$(1,$n$-1) C			 & T$(k - 3 ,k - 2 ;k - 1 )$ 		&$(k + 2)(k - 2)(k - 3)$			 &					\\
	$(a_1,b_4,c_2,d_3,e_{10})$		&(0...10\textbf{1}000...0)						 &$n$(1,$n$-1) C			& -					&$(k + 2)(k)$					 &					 \\
	$(f_5,g_8,h_6,i_7,j_9)$		&(0...1\textbf{1}1001...0)							&1 T, $n$-1 C			& T$(k - 3 ,k ;k + 1 )$						&$(k + 2)(k)(k - 3)$ 				 &					 \\
					&(0...1\textbf{1}1010...0)							&1 T, $n$-1 C			 & T$(k - 2 ,k ;k + 1 )$			&$(k + 2)(k)(k - 2)$				 &					\\
					&(0...1\textbf{1}1011...0)							&$n$(1,$n$-1) C			 &T$(k - 2 ,k - 1 ,k ;k + 1 )$		&$(k + 2)(k)(k - 2)(k - 3)$			 &					\\
					&(0...101\textbf{1}11...0)							&$n$(1,$n$-1) C			 & -					&$(k + 2)(k)(k - 1)(k - 2)(k - 3) $		 &					 \\
					&(0...1\textbf{1}0111...0)							&1 T, $n$-1 C			 & T$(k - 1 ,k' ;k + 1 )$		&$(k + 2)(k - 1)(k - 2)(k - 3) $		 &					\\

\hline

\end{tabular}
}
\end{table*}

%%^^^^^^^^^^^^^^^^^^^^^^^^^^^^^^^^^^^^^^^^^^^^^^^^^^^^^^^^^^^^^^^^^^^^^^^^^^^^^^^^^^^^^^^^^^^^^^^^^^^^^^^^^^^^^^^^^^^^^^^^^^^^^^^^^^^^^^^^^^^^^^^^^^^^^^^^^^

\begin{figure*}
\centerline{\includegraphics[width=5in]{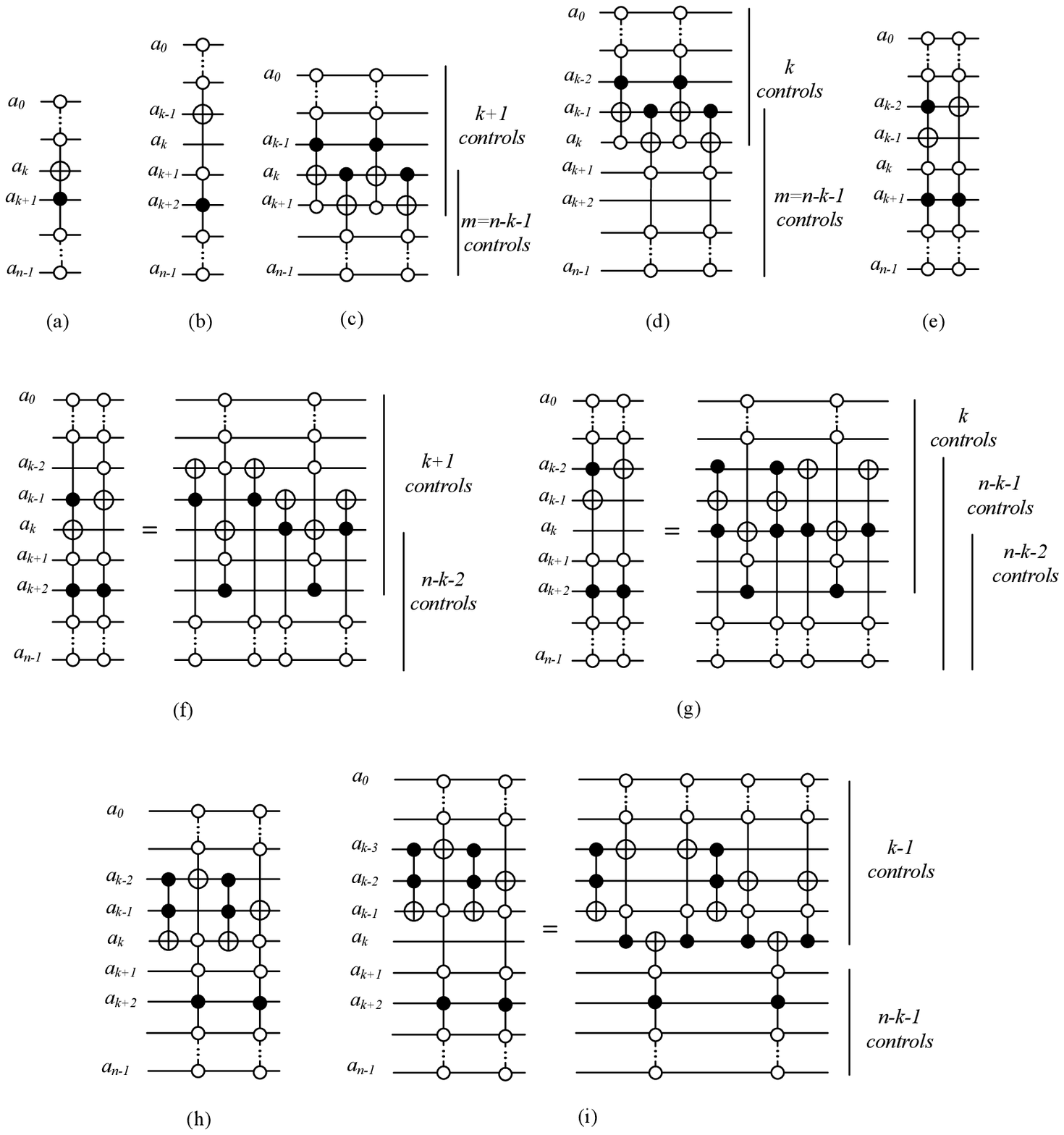}}
\caption{\small The $\kappa_0$ circuit structures for different elementary cycles. The circuit structures for cycles (2,2), (3), (3,3), (4,2), (4,4), (5), and (5,5) are similar to those proposed in \cite{saeedi10}. The new circuits for (2) and (4) besides the application of negative controls and the revised terms in the $\kappa_0$ circuits improve quantum cost and interaction cost.}
\label{fig:all-k}
\end{figure*}

Following the above discussion for the (2,2)-synthesis method, details for the synthesis of other elementary cycles are given in Table \ref{t:EC}. In this table, subscripts in column \emph{Input Cycle(s)} denote orders in considering each term. Intermediate terms are represented by binary expansions with LSB on the right and the underlined bit in the $k$-$th$ position ($k$=$\lfloor \frac{n}{2}\rfloor$). The boldfaced ``1'' is \emph{Pivot} in Algorithm \ref{alg:one} for each term. The parenthesized pairs in column \emph{Max. Cost} represent CNOT count with negative and positive controls, respectively. The numbers given in column \emph{Terms} for the $\kappa_0$ circuit are bit positions with value ``1'' in binary representation. Table \ref{t:tableQC} reports the resulting quantum cost of each elementary cycle.
As can be seen, the total number of elementary gates is improved by a linear factor in most cases.
Considering the worst-case cost of $3(n+\rm{CCL})$ for each gate in the $\pi$ and $\pi^{-1}$ circuits in the LNN architecture and
$6n-12$ elementary gates (i.e., two chains of $n-2$ SWAP gates) for the $\kappa_0$ circuits leads to the results given in \emph{Total Cost (LNN)} column in Table \ref{t:tableQC}.
%%+++++++++++++++++++++++++++++++++++++++++++++++++++++++++++++++++++++
%\input{tableQC.tex}
%%^^^^^^^^^^^^^^^^^^^^^^^^^^^^^^^^^^^^^^^^^^^^^^^^^^^^^^^^^^^^^^^^^^^^^^^^^^^^^^^^^^^^^^^^^^^^^^^^^^^^^^^^^^^^^^^^^^^^^^^^^^^^^^^^^^^^^^^^^^^^^^^^^^^^^^^^^^
% Table
\begin{table*}%
\centering{
\caption{\small Worst-case costs for elementary cycles.}
\label{t:tableQC}
%\scriptsize
\begin{tabular}{|c||clllll||ll|}
\hline

&\multicolumn{6}{|c||}{The Proposed Method}&\multicolumn{2}{|c|}{\cite{saeedi10}}\\
EC& Length & $\kappa_0$ & $\pi$, $\pi^{-1}$ & Total Cost&Cost/Length &Total Cost (LNN) &Total Cost&Cost/Length\\

\hline
\hline
(2)&2&  24$n$-64     & 2$n$+2   & 28$n$-60  &  14$n$-30  	& 145$n^2$-666$n$+772   &34$n$-30	 &17$n$-15   \\
\hline
(2,2)&4&  24$n$-88   & 4$n$+11  & 32$n$-66  &  8$n$-16.5 	& 147$n^2$-791$n$+1100  &34$n$-64	 &8.5$n$-16   \\
\hline
(3)&3&    24$n$-88   & 3$n$+4   & 30$n$-80  &  10$n$-26.7	& 146$n^2$-804$n$+1068  &32$n$-82	 &10.7$n$-27.3   \\
\hline
(3,3)&6&  24$n$-112  & 6$n$+26  & 36$n$-60  &  6$n$-10   	& 149$n^2$-907$n$+1474  &38$n$-46     &6.3$n$-15.3   \\
\hline
(4)&4&  48$n$-152    & 4$n$+11  & 56$n$-130 &  14$n$-32.5	& 291$n^2$-1463$n$+1868 &50$n$-84     &12.5$n$-21   \\
\hline
(4,2)&6&  36$n$-204  & 6$n$+14  & 48$n$-176 & 8$n$-29.4  	& 221$n^2$-1615$n$+2483 &50$n$-122    &8.3$n$-20.3    \\
\hline
(4,4)&8&  36$n$-204  & 8$n$+46  & 52$n$-112 &  6.5$n$-14 	& 223$n^2$-1573$n$+2678 &56$n$-126    &7$n$-15.7   \\
\hline
(5)&5&    48$n$-166  & 5$n$+13  & 58$n$-140 &  11.6$n$-28	& 292$n^2$-1537$n$+2057 &60$n$-130    &12$n$-26   \\
\hline
(5,5)&10& 36$n$-204  & 10$n$+57 & 56$n$-90  &  5.6$n$-9  	& 225$n^2$-319$n$+2790  &64$n$-54     &6.4$n$-5.4   \\
\hline
\end{tabular}
}
\end{table*}%

%%^^^^^^^^^^^^^^^^^^^^^^^^^^^^^^^^^^^^^^^^^^^^^^^^^^^^^^^^^^^^^^^^^^^^^^^^^^^^^^^^^^^^^^^^^^^^^^^^^^^^^^^^^^^^^^^^^^^^^^^^^^^^^^^^^^^^^^^^^^^^^^^^^^^^^^^^^^
%%=============================================================================================================================================================
\subsection{Worst-Case Analysis}\label{subsubsec:WC}

In this section, an upper bound on the number of gates in the proposed cycle-based method is calculated.
To achieve this, let all terms of a truth table be involved in the input cycles to have a cycle with the maximum length $2^n$ for an $n$-input/$n$-output function.
To convert a cycle with length$>$5 to a set of elementary
cycles, we may have some repeated terms in non-disjoint cycles. As such, $2^n$+$a_r$ shows the maximum number of terms where $a_r$ is the maximum number of repeated terms and can be estimated as $a_r=\frac{a_{r-1}+4}{5}, a_0=\frac{2^n}{5}$ which results in $a_r=\frac{2^n}{5} + \sum\nolimits_{i = 2}^{\log_5(\frac{2^n-5}{4})} {\frac{{2^n  + 5^{i}  - 5}}{{5^{i} }}} =  2^{n-2} +  \log_5(\frac{2^n-5}{4}) - \frac{9}{4}$. Theorem \ref{th:wc} discusses the maximum number of elementary gates in our approach.

\begin{theorem}\label{th:wc}
The maximum number of elementary gates for any permutation in the proposed approach is $9.4 n2^n - 18.8 2^n  + o(n^2)$ and $42.4 n^2 2^n  + o(n^3)$ without and with considering interaction cost, respectively.
\end{theorem}

\begin{proof}
In Table \ref{t:tableQC}, the column \emph{Cost/Length} determines a cost needed for setting a term in each elementary cycle. To calculate the maximum cost, suppose at most one 3-cycle, one 4-cycle and one 5-cycle are included which can be synthesized by the related synthesis algorithms. All other terms are supposed to be synthesized as pairs of 2-cycles. Note that the number of elementary gates for fixing terms in a pair of 2-cycles is greater than any other pairs (See Table \ref{t:tableQC}). The repeated terms in non-disjoint 5-cycles are synthesized by the (5,5)-cycle synthesis method.

Accordingly we will have, $3 \times$\emph{Cost/Length${}_3$} + $4 \times$\emph{Cost/Length${}_4$} + $5 \times$\emph{Cost/Length${}_5$} + $(2^n-12) \times$\emph{Cost/Length${}_{2,2}$} + $a_r\times$\emph{Cost/Length${}_{5,5}$} which leads to $9.4 n2^n - 18.8 \times2^n + 2.8 n^2 + 43.5 n -152.1$ elementary gates in the worst-case with arbitrary interaction and $42.4 n^2 2^n  + 11.3 n^3 + 288.2 n^2$ with limited interaction.
\end{proof}

The worst-case quantum cost of \cite{saeedi10} is $51n^2 2^n $ for architectures with limited interaction.

%%=============================================================================================================================================================
\section{Synthesis with Parallel Structure}\label{subsec:parallel}
In this section, a parallel circuit structure is introduced for reversible logic that can be used to considerably reduce circuit depth of reversible circuits in most cases. The general idea is to copy input lines into $k$ sets of zero-initialized ancillae, divide the input specification into $k$ sets of disjoint cycles and then synthesize each set independently by using the prepared ancillae. The final results can be recovered by several CNOTs. It should be mentioned that adding ancillae has been previously used for quantum cost reduction in the synthesis and optimization methods \cite{WilleDAC09,maslov11}. In the proposed method, ancillae are used for the propose of depth reduction without considerable overhead on quantum cost, thanks to the specific form of input representation, i.e., cycle. Note that each cycle can be synthesized by a different synthesis method.

%%=============================================================================================================================================================
\textbf{Input Storing Block.} Copying an arbitrary quantum state is not possible in general but a Boolean value can be copied into a zero-initialized
ancilla by a CNOT gate conditioned on the main line and targeted on the ancilla. For $m$ $n$-line zero-initialized ancillae, the input storing block
includes $m$$n$ CNOT gates with constant depth $m$. Fig. \ref{Pre}-a shows the input storing block for a circuit with $n$ main lines and $m$ $n$-line ancillae. The interaction cost can be calculated as $n(n-1)(1+2+...+m-1)=(1/2)nm(n-1)(m-1)$. Fig. \ref{Pre}-b illustrates another circuit structure with improved interaction cost, $mn(n-1)$. Circuit depth in Fig. \ref{Pre}-a can be improved from linear factor to logarithmic factor $O(\log m)$ \cite{moore09} as shown in Fig. \ref{Pre}-c. Thus, interaction cost can be calculated as $n(n-1)\sum\nolimits_{i = 0}^{\log _2 m - 1} {2^{2i} }=(1/2)n(n-1)(m^{2}-2)$.

\begin{figure*}
\centering
\centerline{\includegraphics[width=4.3in]{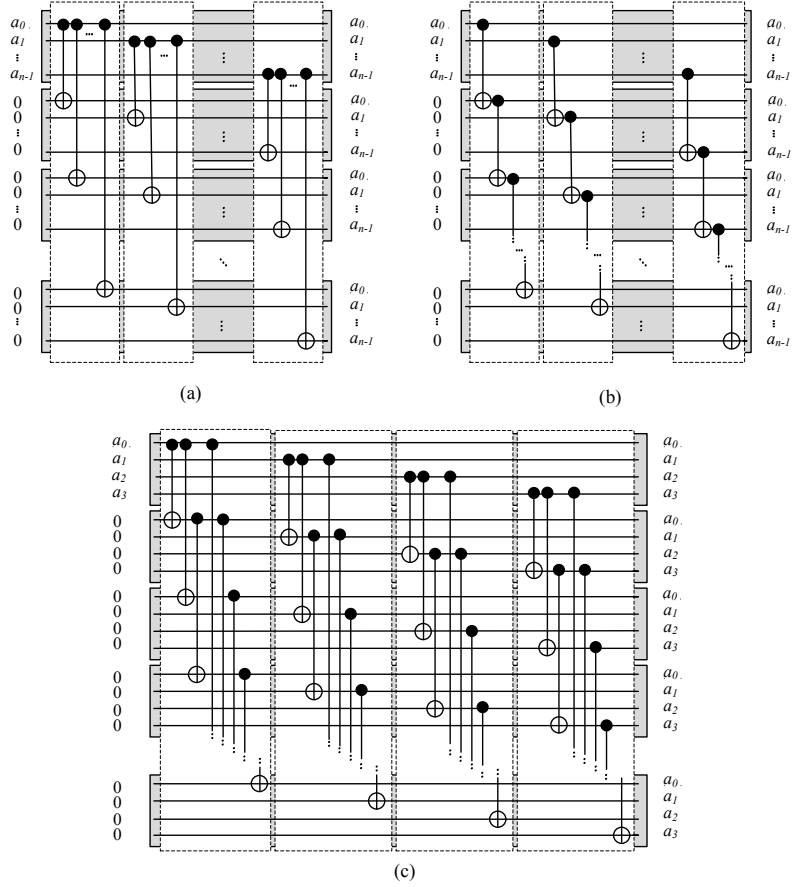}}
\caption{\small(a) The input storing block with linear depth. (b) An alternative circuit structure with improved interaction cost and linear depth. (c) A logarithmic-depth circuit structure.}
\label{Pre}
\end{figure*}

\begin{figure*}
\centering
\centerline{\includegraphics[width=4.3in]{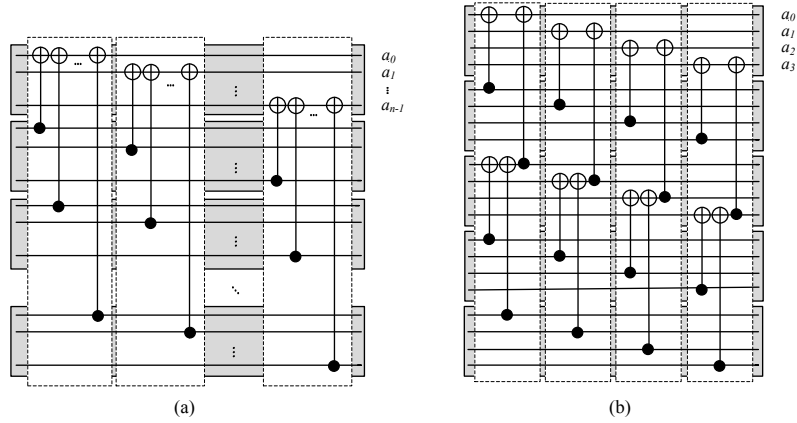}}
\caption{\small(a) The output restoring block with linear depth. (b) The output restoring block with logarithmic depth for four main lines and four 4-line
ancillae.}
\label{Term}
\end{figure*}

\textbf{Output Restoring Block.} Since each subcircuit implements a set of disjoint cycles, for a given input combination, only one circuit (active) produces the results and the outputs of other subcircuits (inactive) are the same as the inputs. The number of inactive subcircuits is equal to the number of $n$-line ancillae registers, which is even. As such, XORing (by CNOT) the outputs of all subcircuits on the main lines cancels inputs and restores correct outputs at the main lines. Overall, for $m$ $n$-line ancillae and $m$+1 sets of disjoint cycles,
$m$$n$ CNOTs with depth $m$ are sufficient. Fig. \ref{Term}-a illustrates the output restoring block for $m$ $n$-line ancillae with interaction cost $nm(n-1)(m-1)$. CNOT-circuit with common target can be implemented with logarithmic depth \cite{moore09} as illustrated in Fig. \ref{Term}-b for $n$=4 and $m$=4. In this case, interaction cost is $n(n-1)
\sum\nolimits_{i = 0}^{\log _2 m - 1} {{{2^i  - 1} \mathord{\left/
 {\vphantom {{2^i  - 1} {2^{i + 1} }}} \right.
 \kern-\nulldelimiterspace} {2^{i + 1} }}}=(1/2)nm(n-1)(2m+\log_2m+2)$.

\begin{theorem}\label{th:par}
Consider a given specification $F$ on $n$ lines written as a set of disjoint cycles $C_1 C_2 ...C_m$ for an odd $m$. Assume that subcircuit $L_i$ implements $C_i$. The specification $F$ can be implemented with depth O($depth_{max}(L_i)$) in the presence of $m$ $n$-line ancillae.
\end{theorem}

\begin{proof}
Copying the input lines to $m-1$ $n$-line zero-initialized ancillae replicates inputs at the ancillae. Disjoint cycles commute. Hence, each subcircuit can be implemented on one register independently.
The input storing/output restoring blocks have constant depth $m$. Therefore, circuit depth is dominated by the maximum depth of all subcircuits.
\end{proof}

\begin{figure*}
\small
\centerline{\includegraphics[width=4.8in]{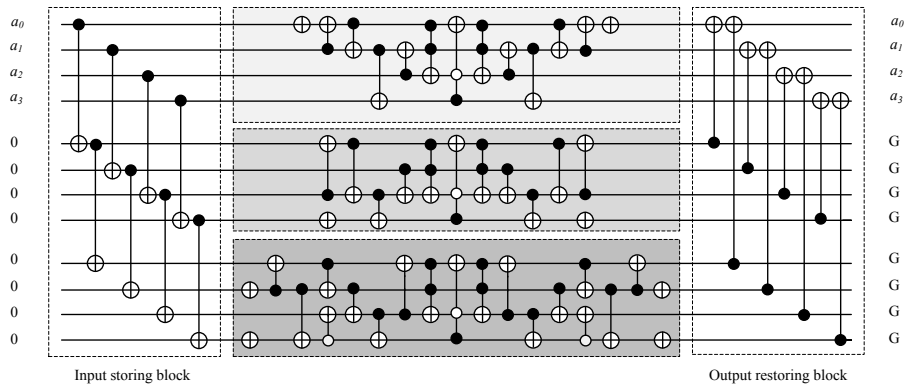}}
\caption{\small An example of the proposed parallel cycle-based structure for a 4-line function.}
\label{fig:exmp}
\end{figure*}

A given specification may contain a set of disjoint cycles with exponential lengths, i.e., $O(2^n)$. In such cases, circuit depth cannot be further improved by Theorem \ref{th:par}. However, as will be shown in Section \ref{sec:exper}, circuit depth can be reduced considerably even with a small number of $n$-line ancillae. To efficiently employ the result of Theorem \ref{th:par}, one needs to determine disjoint cycle sets.

\begin{example}\label{Ex:4}
\emph{Assume that the input cycles (1,3) (7,10) (0,4) (6,15) (2,8) (5,13) are
given for a circuit with 4 lines. All cycles are elementary and no decomposition is required. Let 2 4-line ancillae be available and each pair of 2-cycles be assigned to one set, i.e., (1,3) (7,10)  to set \#1, (0,4) (6,15) to set \#2 and (2,8) (5,13) to set \#3. Applying the input storing block provides the input data on the added zero-initialized ancillae. Now, the proposed method in Section \ref{sec:Mk-cycle} can be applied for each cycle pair which leads to three subcircuits. To combine the results, one needs to add the output restoring block. Accordingly, total depth is equal to the maximum depth of the synthesized subcircuits (i.e., 33) plus 4 (2 for each input storing/output restoring block). Fig. \ref{fig:exmp} illustrates the result.}

\end{example}

%%=============================================================================================================================================================

\textbf{Cycle Distribution.}
Consider $n$ elementary cycles and $m$ register sets, including the input register. The problem is to assign disjoint cycles into different registers such that the total depth of the circuit in each register is minimized and the depths of the registers are almost equal. To achieve this goal, we modeled the cycle distribution problem as the bin packing problem\footnote{Bin packing problem is a combinatorial NP-hard problem in computational complexity theory in which objects of different weights must be packed into a finite number of bins of capacity $W$ such that the number of used bins are minimized.
Given a bin of size $W$ and a list $w_1 ,...,w_n $ of sizes of the items, one should find an integer $B$ and a $B$-partition $S_1  \cup ... \cup S_B $ of $\{ 1,...,n\} $ such that $\sum\nolimits_{i \in S_k } {w_i  \le W} $ for all $k = 1,...,B$. A solution is optimal if it has minimal $B$.
} with a few exceptions. In our modeling, registers are bins and cycles are objects. Each cycle is decomposed into a set of elementary cycles and cost values in Table \ref{t:tableQC} are used as the weights of elementary cycles. If the input permutation is odd, the permutation in one bin should be odd. Many heuristic algorithms have been developed to solve different variants of the bin packing problem. Examples include first fit and best fit algorithms.

To solve the problem, a best fit algorithm is developed which sorts $c$ elementary cycles according to their maximum synthesis costs and
proceeds one cycle at a time. To distribute cycles, the first cycle is
selected and temporarily assigned to bin $i$ for $1\leq i\leq m$. Then, the total cost is calculated among all the bins and the cycle is permanently
assigned to the bin which results in the lowest total cost. In the case of a tie, the bins are selected in sequence. The algorithm
continues until all the cycles are assigned. Therefore, the total time complexity is $O(c\log c) + O(cm^2)$.
At the end, the algorithm
checks the permutation of each bin to make sure that at most one bin has
an odd permutation. Odd permutations need one ancilla in the NCT library
\cite{Shende03}. If more than one bin is found with an odd permutation (called odd bin), the
algorithm moves the smallest odd cycle of the odd bin with maximum
depth to the odd bin with the minimum depth. This can take $O(m)$ time. After the changes, the involved bins should have even
permutations. This process is continued until at most one bin with an
odd permutation exists --- this occurs when the input permutation is odd and at least $m$ even permutations exist to fill all the bins. Altogether, the whole process has a time complexity of $O(c\log c) + O(cm^2)$.
%%=============================================================================================================================================================
%\input{tab.tex}
%%^^^^^^^^^^^^^^^^^^^^^^^^^^^^^^^^^^^^^^^^^^^^^^^^^^^^^^^^^^^^^^^^^^^^^^^^^^^^^^^^^^^^^^^^^^^^^^^^^^^^^^^^^^^^^^^^^^^^^^^^^^^^^^^^^^^^^^^^^^^^^^^^^^^^^^^^^^
\begin{table*}
\centering{
\caption{\small Benchmark specifications before and after decomposition.}
\label{t:tablecycles}
%\scriptsize
\begin{tabular}{|l|c|c||c|c|c|c||c|c|c|c|c|c|c|c|}
\hline
\multirow{2}{*}{Benchmark}	&\multicolumn{9}{c|}{\# of Cycles}				 								&\multicolumn{4}{c|}{\multirow{2}{*}{Depth}}		 \\
\multirow{2}{*}{Function}	&\multicolumn{2}{c||}{Before DCM} &\multicolumn{4}{c||}{After DCM} &	\multicolumn{3}{c|}{After DCM \& DIST}			 &\multicolumn{4}{c|}{}	      				\\
				&EC & 	nEC & (2) & (3) & (4) & (5)		 & 	set1 & set2 & set3 				& circ1 & circ2 & circ3 & total		\\
\hline																							
\hline
hwb8				&48	&16	&36	&-	&28	&16			&26	&28	&26				&1923	&1995	&1953	&1999	\\					
\hline												
hwb9				&54	&38	&-	&54	&-	&76			&43	&43	&44				&4344	&4347	&3988	&4351	\\				
\hline												
hwb10				&228	&26	&152	&-	&-	&154			&101	&103	&102				&9929	&10058	&9898	&10062	\\				
\hline												
hwb11				&-	&186	&-	&186	&-	&372			&186	&186	&186				&23862	&23826	&23827	&23866	\\				
\hline												
nth\_prime7 			&-	&5	&1	&3	&1	&28			&19	&6	&8				&1519	&390	&734	&1523	\\				
\hline												
nth\_prime8 			&-	&1	&-	&1	&-	&62			&63	&-	&-				&5852	&-	&-	&5852	\\													
\hline												
nth\_prime9			&1	&3	&2	&-	&1	&125			&122	&4	&2				&13783	&393	&346	&13787	\\	 			
\hline												
nth\_prime10			&1	&3	&1	&2	&-	&253			&96	&85	&75				&12115	&10947	&9329	&12119	\\	 		 	
\hline												
nth\_prime11			&-	&6	&2	&-	&1	&507			&315	&36	&159				&46888	&5470	&23765	&46892	\\	 		 	
\hline
\end{tabular}
}
\end{table*}
%%^^^^^^^^^^^^^^^^^^^^^^^^^^^^^^^^^^^^^^^^^^^^^^^^^^^^^^^^^^^^^^^^^^^^^^^^^^^^^^^^^^^^^^^^^^^^^^^^^^^^^^^^^^^^^^^^^^^^^^^^^^^^^^^^^^^^^^^^^^^^^^^^^^^^^^^^^^

\section{Experimental Results}\label{sec:exper}
The proposed cycle-based synthesis method for the LNN architecture and the suggested parallel structure for reversible logic synthesis were implemented in C++ and all of the experiments were performed on an Intel Pentium IV 2.5GHz computer with 4GB memory.
To evaluate the proposed synthesis method, some of the reversible benchmark functions from \cite{MaslovSite} were synthesized. The selection criteria for these benchmarks will be discussed later in this section and their specifications are given in Table \ref{t:tablecycles} \emph{before} and \emph{after} decomposition. The decomposition approach of \cite{saeedi10} is used in our method to decompose the input cycles into the proposed elementary cycles. The number of elementary cycles (EC) and non-elementary cycles (nEC) of each benchmark is reported is this table. After decomposition, all cycles are elementary with length$<$6.
Note that \cite{saeedi10} proposes the best prior synthesis algorithm for medium-size \emph{hwbN} and \emph{N-th prime} functions if no ancilla is available \cite{MaslovSite}. While \emph{hwbN} functions can be implemented with a polynomial cost $O(n \log^2 n)$ if a logarithmic number of garbage bits $\lceil \log n\rceil + 1$ is available \cite{MaslovSite}, the proposed approach is more general and can be applied to many reversible functions.

To evaluate the proposed parallel structure, the cycle-based algorithm of Section \ref{sec:Mk-cycle} was used for synthesizing each subset. Since the number of signals is limited in the current quantum technologies, the minimum number of ancillae (2 $n$-line registers) was used. Therefore, the number of input cycles should be $>$3 to have at least one cycle in each subset.
In our experiments, the results of \cite{MaslovTCAD08} were used for decomposing multiple-control Toffoli gates and calculating quantum cost for the gates with negative controls. Besides, the two-qubit cost model of \cite{maslov11} is used for evaluating the results.
A naive SWAP insertion method and the method of \cite{saeedi11QIP} were used to evaluate the results for the LNN architecture. For the naive method, \emph{move} and \emph{delete} rules were applied on the synthesized circuits to remove redundant gates.
To estimate circuit depth, the greedy level compaction algorithm of \cite{MaslovTCAD08} was implemented without applying the templates.
%\input{tableRes1.tex}
%\input{tableRes4.tex}
%%^^^^^^^^^^^^^^^^^^^^^^^^^^^^^^^^^^^^^^^^^^^^^^^^^^^^^^^^^^^^^^^^^^^^^^^^^^^^^^^^^^^^^^^^^^^^^^^^^^^^^^^^^^^^^^^^^^^^^^^^^^^^^^^^^^^^^^^^^^^^^^^^^^^^^^^^^^
\begin{table}
\caption{\small Comparison of the proposed approach and prior best results. \#A is the number of ancillae. R and P are used for regular and parallel structures, respectively. The resulted circuits are available at
http://ceit.aut.ac.ir/\~{}arabzadeh/results/, and may be viewed with RCViewer+ \cite{RCV}.}
\label{t:tableRes1}
%\scriptsize
\begin{center}
\begin{tabular}{|l|c|c|c|c|c|c|c|c|c|c|c|c|}
\hline
Benchmark			&\multirow{2}{*}{$n$}	&\multicolumn{5}{c|}{The Proposed Method}		& \multicolumn{3}{c|}{\cite{saeedi10}}	& \multicolumn{3}{c|}{Improvement (\%)} 	\\
Function			&			&R/P 	&\# A	&QC	& 2-qubit	&Depth 		&QC	& 2-qubit	& Depth  		 &QC		& 2-qubit	&Depth  	\\
\hline
\hline
\multirow{2}{*}{hwb8}		&\multirow{2}{*}{8}	&R	&-	&6686	&4468	&5622	 &\multirow{2}{*}{6940}		&\multirow{2}{*}{5348}	&\multirow{2}{*}{5442}			 &3.6	 &16.4	&-3.3	 \\
				&			&P	&16	&6964	&4730	&1999	&				&			&			 		&-0.3	&11.5	&63.2	\\
\hline
\multirow{2}{*}{hwb9}		&\multirow{2}{*}{9}	&R	&-	&14474	&10382	&12054	 &\multirow{2}{*}{16173}	&\multirow{2}{*}{12479}	&\multirow{2}{*}{12472}			 &10.5	 &16.8	&3.3	 \\
				&			&P	&18	&15262	&10764	&4351	&				&			&			 		&5.6	&13.7	&65.1	\\
\hline			
\multirow{2}{*}{hwb10}		&\multirow{2}{*}{10}	&R	&-	&35298	&23584	&29751	 &\multirow{2}{*}{35618}	&\multirow{2}{*}{25453}	&\multirow{2}{*}{27812}			 &0.8	 &7.3	&-6.9	 \\
				&			&P	&20	&35890	&23874	&10062	&				&			&			 		&-0.7	&6.2	&63.8	\\
\hline														
\multirow{2}{*}{hwb11}		&\multirow{2}{*}{11}	&R   	&-	&86864	&65260	&71418	 &\multirow{2}{*}{90745}	&\multirow{2}{*}{71175}	 &\multirow{2}{*}{69763}	 		&4.2	&8.3	&-2.3	 \\
	 			&			&P	&22	&87234	&65442	&23866	&				&			&			 		&3.8	&8.0	&65.7 	\\
\hline
\multirow{2}{*}{nth\_prime7} 	&\multirow{2}{*}{7}	&R	&-	&2888	&2296	&2473	 &\multirow{2}{*}{3172}		&\multirow{2}{*}{2841}	&\multirow{2}{*}{2514}			 &8.9	&19.1	&1.6	 \\
				&			&P	&14	&3100	&2398	&1523	&				&			&			 		&2.2	&15.5	&39.4   \\
\hline
\multirow{2}{*}{nth\_prime8} 	&\multirow{2}{*}{8}	&R	&-	&7016	&5624	&5852	 &\multirow{2}{*}{7618}		&\multirow{2}{*}{6622}	&\multirow{2}{*}{5793}			 &7.9 	&15.0	&-1.0   \\
				&			&P	&-	&-	&-	&-	&				&			&					&-	&-	&-	 \\											
\hline														
\multirow{2}{*}{nth\_prime9}	&\multirow{2}{*}{9}	&R   	&-	&16820	&11907	&14285	 &\multirow{2}{*}{17975}	&\multirow{2}{*}{14076}	 &\multirow{2}{*}{13941}	 		&6.4	&15.4	&-2.4	 \\
	 			&			&P	&18	&17507	&12053	&13787	&				&			&			 		&2.6	&14.3	&1.1  	\\
\hline
\multirow{2}{*}{nth\_prime10}	&\multirow{2}{*}{10}	&R	&-	&38843	&27743	&31924	 &\multirow{2}{*}{40301}	&\multirow{2}{*}{31841}	 &\multirow{2}{*}{31254}	 		&3.6	&12.8	&-2.1   \\
	 		 	&			&P	&20	&39317	&27933	&12119	&				&			&			 		&2.4	&12.2	&61.2   \\
\hline														
\multirow{2}{*}{nth\_prime11}	&\multirow{2}{*}{11}	&R   	&-	&92863	&67401	&75668	 &\multirow{2}{*}{95433}	&\multirow{2}{*}{75474}	 &\multirow{2}{*}{72934}	 		&2.6	&10.6	&-3.7	 \\
	 			&			&P	&22	&93389	&67677	&46892	&				&			&			 		&2.1	&10.3	&35.7 	\\
\hline
\hline
\multicolumn{2}{|l}{\multirow{2}{*}{\textbf{Average}}} 	& \multicolumn{8}{c|}{} &\textbf{5.4}	&\textbf{13.6}  	&\textbf{-1.9}\\
\multicolumn{2}{|l}{}		 			& \multicolumn{8}{c|}{} &\textbf{2.2} 	&\textbf{11.5}   	&\textbf{49.4}  \\
\hline
\end{tabular}
\end{center}
\end{table}
\begin{table*}
\caption{\small Comparison of the proposed approach and the one in \cite{saeedi10} with the nearest neighbor limitation. The improvment column compares the results after applying \cite{saeedi11QIP} on both methods. The resulted circuits are available at
http://ceit.aut.ac.ir/\~{}arabzadeh/results/, and may be viewed with RCViewer+ \cite{RCV}.}
\label{t:tableRes4}
%\scriptsize
\centering{
\begin{tabular}{|l|c|c|c|c|c|c|c|c|c|c|c|c|c|}
\hline
\multirow{2}{*}{Benchmark}	&\multirow{3}{*}{$n$}	&\multicolumn{6}{c|}{The Proposed Method}	 																										 &\multicolumn{2}{c|}{\multirow{2}{*}{\cite{saeedi10}+\cite{saeedi11QIP}}}	& \multicolumn{2}{c|}{\multirow{2}{*}{Improvement (\%)}}  	\\
\multirow{2}{*}{Function}	&			&\multirow{2}{*}{R/P} 				&\multirow{2}{*}{\# A}	 			&\multicolumn{2}{c|}{+Naive}  				 &\multicolumn{2}{c|}{+\cite{saeedi11QIP}}  			&\multicolumn{2}{c|}{}								 &\multicolumn{2}{c|}{}						\\
				&			&						&						&QC		&Depth					&QC	&Depth								 &QC	 &Depth	 								&QC	&Depth 							\\
\hline
\hline
\multirow{2}{*}{hwb8}		&\multirow{2}{*}{8}	&R	&-	&36684		&32313		&31553	&20940		&\multirow{2}{*}{36732}	&\multirow{2}{*}{22720}		&14.0	 &7.8		\\		
				&			&P	&16	&46788		&14758		&36045	&9248		&			 &				&1.8	&59.2		\\		
\hline																			
\multirow{2}{*}{hwb9}		&\multirow{2}{*}{9}	&R	&-	&87310		&74676		&77860	&46958		&\multirow{2}{*}{91805}	&\multirow{2}{*}{51181}		&15.1	 &8.2		\\		
				&			&P	&18	&100228		&31810		&87389	&19597		&			 &				&4.8	&61.7		\\		
\hline																			
\multirow{2}{*}{hwb10}		&\multirow{2}{*}{10}	&R	&-	&279496		&248524		&202903	&112623		&\multirow{2}{*}{228240}&\multirow{2}{*}{117893	 }	 &11.1	&4.4		\\	
				&			&P	&20	&291014		&89021		&212616	&41479		&			 &				&6.8	&64.8		\\		
\hline																
\multirow{2}{*}{hwb11}	&\multirow{2}{*}{11}		&R   	&-	&682182		&605294		&562817	&297986		&\multirow{2}{*}{611843} &\multirow{2}{*}{307114}	 &8.0	&2.9	 	\\
	 			&			&P	&22	&685944		&205472		&569876	&104372		&			 &				&6.8	&66.0 		\\
\hline																					
\multirow{2}{*}{nth\_prime7} 	&\multirow{2}{*}{7}	&R	&-	&12264		&10649		&10922	&9799		&\multirow{2}{*}{15356}	 &\multirow{2}{*}{10130}	 &28.8	&3.2		 \\		
				&			&P	&14	&15106		&7734		&15897	&6930		&			 &				&-3.5	&31.5	   	\\		
\hline	
\multirow{2}{*}{nth\_prime8} 	&\multirow{2}{*}{8}	&R	&-	&35976		&29975		&30796	&26920		&\multirow{2}{*}{42059}	 &\multirow{2}{*}{24574}	 &26.7	&-9.5	  	 \\
				&			&P	&-	&-		&-		&-	&-		&			&				&-	&-		\\											
\hline														
\multirow{2}{*}{nth\_prime9}	&\multirow{2}{*}{9}	&R   	&-	&91984		&76910		&90511	&54457		&\multirow{2}{*}{99003}	 &\multirow{2}{*}{55737}	 &8.5	&2.2		 \\
	 			&			&P	&18	&98686		&76020		&95362	&54850		&			 &				&3.6	&1.5	 	\\
\hline
\multirow{2}{*}{nth\_prime10}	&\multirow{2}{*}{10}	&R	&-	&241538		&199996		&222865	&124122		&\multirow{2}{*}{248901}&\multirow{2}{*}{137091}	 &10.4	&9.4	  	 \\
	 		 	&			&P	&20	&250526		&79165		&228777	&49613		&			 &				&8.0	&63.8	   	\\
\hline															
\multirow{2}{*}{nth\_prime11}	&\multirow{2}{*}{11}	&R   	&-	&654910		&577721		&576047	&308413		&\multirow{2}{*}{625320} &\multirow{2}{*}{324005}	&7.8	&4.8		 \\
	 			&			&P	&22	&665132		&361756		&585165	&195500		&			 &				&6.4	&39.6	 	\\
\hline
\hline
\multicolumn{2}{|l}{\multirow{2}{*}{\textbf{Average}}} 	& \multicolumn{8}{c|}{} &   	 											\textbf{14.6}   &\textbf{3.8}	\\
\multicolumn{2}{|l}{}		 			& \multicolumn{8}{c|}{} &	 											\textbf{4.4} 	&\textbf{48.6} 	\\
\hline
\end{tabular}
}
\end{table*}		
%%^^^^^^^^^^^^^^^^^^^^^^^^^^^^^^^^^^^^^^^^^^^^^^^^^^^^^^^^^^^^^^^^^^^^^^^^^^^^^^^^^^^^^^^^^^^^^^^^^^^^^^^^^^^^^^^^^^^^^^^^^^^^^^^^^^^^^^^^^^^^^^^^^^^^^^^^^^
Table \ref{t:tableRes1} and Table \ref{t:tableRes4} report the quantum cost (QC), the two-qubit cost (2-qubit) and the depth (Depth) for the synthesized circuits without and with limited interaction.
Since \cite{saeedi10} does not target the limited interaction in the LNN architecture, we used the method of \cite{saeedi11QIP} on the results of \cite{saeedi10} and ours to insert SWAP gates.
Runtime of \cite{saeedi10} and our method is less than one minute for the selected benchmarks. In the proposed method,
this time includes the time required for applying the distribution procedure in the parallel structure and the
time required for synthesis and applying the \emph{move} and \emph{delete} rules.
In the parallel structure, due to the qubit reordering in \cite{saeedi11QIP}, at most $3n(3n-1)$ SWAP gates are used between the input storing block, the subsets and the output restoring block to order lines.

As can be seen in Table \ref{t:tableRes4}, the effect of the post-process method is more significant for \cite{saeedi10} but altogether the results of the proposed LNN-based method are better than those of \cite{saeedi10} after applying \cite{saeedi11QIP} on both methods.
Notice that using negative controls does not allow to increase the quantum cost.  For odd permutations, one more ancilla should be added.
The two-qubit costs are compared in Table \ref{t:tableRes1} and the results show 13.6\% and 11.5\% improvement on average for the regular and parallel structures, respectively.
In the parallel structure, the average depth improvement of the \emph{N-th prime} benchmarks is less than that of \emph{hwbN} functions since the input cycles of those functions are unstructured with different cycle lengths which result in unbalanced subsets after distribution.
Input cycle distributions after decomposition (DCM) and distribution (DIST) are reported in Table \ref{t:tablecycles}.
For \emph{hwbN} functions, applying the distribution method leads to 3 sets with almost the same numbers of elementary cycles.
We report the circuit depth for each set along with the total depth after considering the effect of input storing and output restoring blocks in this table.
As reported in Table \ref{t:tablecycles}, function \emph{nth\_prime8} has one disjoint input cycle. Accordingly, the resulting elementary cycles should be assigned to one set by the proposed method.

In choosing the benchmark functions that were considered in this paper, the general guidelines presented in \cite{saeedi10} and \cite{saeedi2012} were considered. These guidelines stipulate that one of the scenarios in which the cycle-based methods render significantly superior results is when the input function contains permutations without regular patterns such as \emph{hwbN}, \emph{N-th prime} \cite{saeedi10} functions. For this reason, only the results of these functions are reported in this paper.
As for other functions in \cite{MaslovSite}, some are reported in \cite{saeedi10} along with a discussion on their suitability for the cycle-based approach (like Permanent). To avoid being repetitive, we did not include this set in this paper. There are yet other benchmarks that include important arithmetic functions like adders, multipliers and group arithmetic (e.g., in Galois Fields). Since the proposed cycle-based synthesis method is a general synthesis approach, it may not produce interesting results compared to other approaches specifically developed for those benchmark functions.

%%%%%%%%%%%%%%%%%%%%%%%%%%%%%%%%%%%%%%%%%%%%%%%%%%%%%%%%%%%%%%%%%%%%%%%%%%%%%%%%%%%%%%%%%%%%%%%%%%%%%%%%%%%%%%%%%%%%%%%%%%%%%%%%%%%%%%%%%%%%%%%%%%%%%%%%%%%%%%%%%%
\section{Conclusion} \label{sec:conclu}
In this paper, a synthesis approach is proposed in order to reduce logical depth for architectures with limited interactions which applies a cycle-based approach to synthesize a given specification. The proposed method focuses on the interaction cost and depth besides the traditional quantum cost metric as a multi-objective view in the large picture. To achieve this, we redesigned the elementary cycles in \cite{saeedi10} with negative controls and limited interaction between gate lines. Moreover, a new parallel circuit structure was proposed for reversible logic in the presence of several ancillae registers. Altogether, the mentioned structure, which can be used with other synthesis methods, filling with the proposed cycle-based synthesis method for interaction cost leads to our whole flow for depth-optimized reversible circuit synthesis.

A given permutation is written as a set of disjoint cycles to be used in the proposed parallel circuit structure. Then, the resulting cycles are distributed among the available $n$-line registers based on the bin packing problem. The cycles are then synthesized on the assigned registers independently.
Our experiments and analysis show the effectiveness of the proposed approach with and without the interaction cost limitations for the attempted benchmarks and in the worst-case.

%Furthermore, considering the proposed method for other quantum architectures can lead to various considerations in some parts of the proposed method. Finding distinct subsets in quantum logic and using similar parallelism could be interesting.
%Trading-off discussion of the number of input qubits and the circuit depth in more practical manner, is of the future works.

%%%%%%%%%%%%%%%%%%%%%%%%%%%%%%%%%%%%%%%%%%%%%%%%%%%%%%%%%%%%%%%%%%%%%%%%%%%%%%%%%%%%%%%%%%%%%%%%%%%%%%%%%%%%%%%%%%%%%%%%%%%%%%%%%%%%%%%%%%%%%%%%%%%%%%%%%%%%%%%%%%
\bibliographystyle{unsrt}

\begin{thebibliography}{10}

\bibitem{markov2012}
I.~L. Markov and M.~Saeedi.
\newblock Constant-optimized quantum circuits for modular multiplication and
  exponentiation.
\newblock {\em Quant. Inf. and Comput.}, 12(5\&6):0361--0394, 2012.

\bibitem{Aaronson&Gottesman04}
S.~Aaronson and D.~Gottesman.
\newblock Improved simulation of stabilizer circuits.
\newblock {\em Phys. Rev. A}, 70:052328, 2004.

\bibitem{saeedi2012}
M.~Saeedi and I.~L. Markov.
\newblock Synthesis and optimization of reversible circuits - a survey.
\newblock {\em ACM Computing Surveys, e-print, arXiv:1110.2574}, 2012.

\bibitem{cheung07}
D.~Cheung, D.~Maslov, and S.~Severini.
\newblock Translation techniques between quantum circuit architectures.
\newblock In {\em Workshop on Quantum Information Processing}, 2007.

\bibitem{meter06}
R.~Van Meter and M.~Oskin.
\newblock Architectural implications of quantum computing technologies.
\newblock {\em J. Emerg. Technol. Comput. Syst.}, 2(1):31--63, 2006.

\bibitem{maslov11}
D.~Maslov and M.~Saeedi.
\newblock Reversible circuit optimization via leaving the \uppercase{B}oolean
  domain.
\newblock {\em IEEE Trans. on CAD}, 30(6):806--816, 2011.

\bibitem{Gupta06}
P.~Gupta, A.~Agrawal, and N.~K. Jha.
\newblock An algorithm for synthesis of reversible logic circuits.
\newblock {\em IEEE Trans. on CAD}, 25(11):2317--2330, 2006.

\bibitem{Maslov07}
D.~Maslov, G.~W. Dueck, and D.~M. Miller.
\newblock Techniques for the synthesis of reversible {Toffoli} networks.
\newblock {\em ACM Trans. Des. Autom. Electron. Syst.}, 12(4):42, 2007.

\bibitem{WilleDAC09}
R.~Wille and R.~Drechsler.
\newblock \uppercase{BDD}-based synthesis of reversible logic for large
  functions.
\newblock {\em Design Autom. Conf.}, pages 270--275, 2009.

\bibitem{saeedi10}
M.~Saeedi, M.~Saheb~Zamani, M.~Sedighi, and Z.~Sasanian.
\newblock Reversible circuit synthesis using a cycle-based approach.
\newblock {\em J. Emerg. Technol. in Comput. Syst.}, 6(4):1--26, December 2010.

\bibitem{miller10}
D.~M. Miller, R.~Wille, and R.~Drechsler.
\newblock Reducing reversible circuit cost by adding lines.
\newblock {\em Int'l Symp. on Multiple-Valued Logic}, pages 217--222, 2010.

\bibitem{saeedi11QIP}
M.~Saeedi, R.~Wille, and R.~Drechsler.
\newblock Synthesis of quantum circuits for nearest neighbor architectures.
\newblock {\em Quant. Inf. Proc.}, 10(3):355--377, 2011.

\bibitem{hirata2011}
Y.~Hirata, M.~Nakanishi, S.~Yamashita, and Y.~Nakashima.
\newblock An efficient conversion of quantum circuits to a linear nearest
  neighbor architecture.
\newblock {\em Quant. Inf. and Comput.}, 11(1\&2):0142--0166, 2011.

\bibitem{MaslovTCAD08}
D.~Maslov, G.~W. Dueck, D.~M. Miller, and C.~Negrevergne.
\newblock Quantum circuit simplification and level compaction.
\newblock {\em IEEE Trans. on CAD}, 27(3):436--444, March 2008.

\bibitem{moore09}
C.~Moore and M.~Nilsson.
\newblock Parallel quantum computation and quantum codes.
\newblock {\em SIAM Journal on Computing}, 31:799--815, 2001.

\bibitem{Shende03}
V.~V. Shende, A.~K. Prasad, I.~L. Markov, and J.~P. Hayes.
\newblock Synthesis of reversible logic circuits.
\newblock {\em IEEE Trans. on CAD}, 22(6):710--722, June 2003.

\bibitem{SaeediMEJ10}
M.~Saeedi, M.~Sedighi, and M.~{Saheb Zamani}.
\newblock A library-based synthesis methodology for reversible logic.
\newblock {\em Microelectron. J.}, 41(4):185--194, Apr 2010.

\bibitem{MaslovSite}
D.~Maslov.
\newblock Reversible logic synthesis benchmarks page.
\newblock {\em http://webhome.cs.uvic.ca/\~{}dmaslov}, 2011.

\bibitem{RCV}
M.~Arabzadeh, and M.~Saeedi.
\newblock RCviewer+, A viewer/analyzer for reversible and quantum circuits, version 2.41. 
\newblock {\em available at http://ceit.aut.ac.ir/QDA/RCV.htm}, 2011.


\end{thebibliography}

\end{document}